\documentclass[12pt,reqno]{amsart}
\usepackage[samelinetheorem,notitle]{maherart}

\usepackage{xurl}
\usepackage{booktabs}
\usepackage{tabularx}
\usepackage{amsfonts,enumerate,bm,bbm,paralist,nicefrac,comment,float,thmtools,thm-restate,comment,multirow}
\usepackage[export]{adjustbox}
\makeatletter
\newcommand\fs@boxedtop
  {\fs@boxed
   \def\@fs@mid{\vspace\abovecaptionskip\relax}%
   \let\@fs@iftopcapt\iftrue
  }
\makeatother
\floatstyle{boxedtop}
\floatname{framedbox}{Procedure}
\newfloat{framedbox}{tbp}{lob}

\hypersetup{
    pdftitle =
        {The centralizing effects of private order flow on proposer builder separation},
    pdfauthor =
        {Tivas Gupta, Mallesh M Pai and Max Resnick},
    pdfsubject=
        {OFAs + PBS is a Bad Combination}
}

\usepackage{xparse}

\DeclareDocumentCommand\Pr{ m g }{%
    \ensuremath{   \IfNoValueTF {#2}
      {\mathbb{P}\left[{#1}\right]}
      {\mathbb{P}\left[{#1}\middle\vert{#2}\right]}%
    }
}
\DeclareDocumentCommand\E{ m g }{%
    \ensuremath{   \IfNoValueTF {#2}
      {\mathbb{E}\left[{#1}\right]}
      {\mathbb{E}\left[{#1}\middle\vert{#2}\right]}%
    }
}

\newif\ifdraft
\drafttrue

\begin{document}
\raggedbottom
\title{The centralizing effects of private order flow on proposer-builder separation}
\author[Gupta]{Tivas Gupta$^\text{A}$}
\address{$^\text{A}$ Special Mechanisms Group,
\href{mailto:tivas.gupta@mechanism.org}{tivas.gupta@mechanism.org}}
\author[Pai]{Mallesh M. Pai$^\text{B}$}
\address{$^\text{B}$ Rice University and Special Mechanisms Group
\\
\href{mailto:mallesh.pai@rice.edu}{mallesh.pai@rice.edu}}
\author[Resnick]{Max Resnick $^\text{C}$}
\address{$^\text{C}$ Special Mechanisms Group
\href{mailto:max.resnick@mechanism.org}{max.resnick@mechanism.org}}
\address{\today}
\begin{abstract}
The current Proposer-Builder Separation (PBS) equilibrium has several builders with different backgrounds winning blocks consistently. This paper considers how that equilibrium will shift when transactions are sold privately via order flow auctions (OFAs) rather than forwarded directly to the public mempool. We discuss a novel model that highlights the augmented value of private order flow for integrated builder searchers. We show that private order flow is complementary to top-of-block opportunities, and therefore integrated builder-searchers are more likely to participate in OFAs and outbid non integrated builders. They will then parlay access to these private transactions into an advantage in the PBS auction, winning blocks more often and extracting higher profits than non-integrated builders. To validate our main assumptions, we construct a novel dataset pairing post-merge PBS outcomes with realized 12-second volatility on a leading CEX (Binance). Our results show that integrated builder-searchers are more likely to win in the PBS auction when realized volatility is high, suggesting that indeed such builders have an advantage in extracting top-of-block opportunities. Our findings suggest that modifying PBS to disentangle the intertwined dynamics between top-of-block extraction and private order flow would pave the way for a fairer and more decentralized Ethereum. 


\end{abstract}


\maketitle

\section{Introduction}\label{sec:intro}
Most Ethereum blocks today are built by specialized \textit{builders} rather than validators. In every slot, builders gather transactions and assemble them into blocks. They then compete against each other in an ascending price (English) auction for the right to have the block they assembled proposed by the proposer. Whichever builder bids the highest wins the Proposer-Builder Separation (PBS) auction, and pays their bid to the proposer.  

The right to build a block is valuable for several reasons, most obviously because users pay \textit{tips} for inclusion. Presently these tips make up only a small portion of the total value from building a block. A majority of the value from building a block comes from the builder exploiting \emph{MEV opportunities}. MEV (Maximal Extractable Value) refers to additional value that can be exploited from strategically reordering or including specific transactions. 

Current MEV opportunities on Ethereum can be broadly segmented into two categories: \textit{top-of-block} and \textit{block body}. Let us describe each in turn. Top-of-block opportunities are primarily CEX/DEX arbitrage: exploiting price divergences of a token between a centralized exchange (CEX) and some on-chain Decentralized Exchange (DEX) operated by a smart contract, e.g., Uniswap. Intuitively, successfully exploiting such a price divergence requires both priority access to the first few transactions in the block on-chain, and also high quality execution on the centralized exchange. The latter requires high-frequency trading (HFT) strategies and low CEX transaction fees.

Block-body opportunities are typically frontrunning attacks that involve sandwiching user transactions or executing user orders against each other to cut out the liquidity providers. The value of the Block-body is primarily dictated by access to transactions. Historically, most transactions have been forwarded to the public mempool, meaning all block builders have access to the same transactions; however, some builders have access to private order flow which is not available in the public mempool. The availability of private order flow is likely to be further supplemented in the near future by the advent of order flow auctions (OFAs), venues where order flow providers (wallets) sell the exclusive right to execute their users' transactions. 

This paper focuses on the complementarity between top-of-block and block body opportunities. In particular, the PBS auction makes no distinction between top-of-block and block body, instead, the right to build the entire block is sold wholesale. This means that an advantage in top-of-block extraction capability can help secure value from the body of the block and vice versa.


This paper makes two contributions. First, we demonstrate empirically that builders operated by high-frequency trading firms are superior at capturing the top of block opportunities. Second, we construct a simple model of proposer-builder separation and demonstrate that, in this model, private order flow is more valuable to vertically integrated builder searchers than non-integrated builders.  Our theoretical results therefore imply that private order flow markets are likely to be dominated by these firms. 

Let us now describe our analysis and results in a little more detail. The main assumption in our subsequent theoretical analysis is that some bidders are stochastically advantaged at extracting top-of-block opportunities. We validate this assumption empirically. In particular, we construct a unique dataset that combines roughly a month PBS auction outcome data, i.e., which builder won which blocks over the course of a month; paired with detailed price data on a major CEX, namely, Binance. Our empirical strategy posits that the realized 12-second volatility of ETH on Binance is plausibly exogenous. Therefore the realized volatility will generate blocks that have varying top of block value. A block in which the price on Binance is flat over the previous 12 seconds will have almost no top of block value, meaning any advantage that builder searchers have at extracting from the top of the block will be irrelevant. In contrast, if the price shift is large in that period, winning the block becomes should be far more valuable to builders who excel at top-of-block extraction. Our results show that when absolute log price change on Binance in a 12 second period is large, builder-searchers operated by HFTs are far more likely to win. These results are statisitcally significant, invalidating the null hypothesis that all builders are roughly equivalent in their top-of-block extraction capability.  

Having demonstrated this, we turn to a theoretical model which explores the centralizing effects of this top-of-block advantage on the equilibrium of the PBS auction. Our model considers a simple abstraction where a block can contain at most two transactions. 

In the first stage of our model, the builders gather block body opportunities. We consider two scenarios. In the first, this transaction is sourced from the public mempool, i.e. all builders have (free) access to the block-body transactions. This models relatively well the current state of affairs. The second scenario envisages builders purchasing transactions in an OFA, which models the plausible scenario we are moving towards. In the second stage of our model, the builders combine their block body transactions with their top-of-block transactions to form a block and then compete with each other in an English (i.e., ascending) auction for the right to append their block onto the chain.

Our results show that advantages in top of block extraction capabilities are magnified when private order flow is available, in comparison with the current scenario where block body opportunities are available in the public mempool or otherwise shared with all builders. In particular, a builder with an advantage, be it deterministic or stochastic, at extracting top-of-block opportunities, will win the OFA. With access to the private transactions, it will then win the PBS auction more often, and have higher profits, than it would have in the counterfactual world without OFAs/ private transactions.

Our results suggest a troubling centralizing tendency of PBS when private order flow is available via an Order Flow Auction, a setting we are moving towards. In particular, a small number of integrated builder-searchers who have top-of-block extraction capabilities will dominate both the OFAs and the downstream PBS auction. This contrasts with the popular idea that the OFAs and the PBS auction will squeeze proposer profits between the validators (who earn the PBS auction revenue) and the order flow providers (who earn the OFA revenue). This also contrasts with the original goal of PBS which was to keep block building decentralized. 

Our results therefore provide a further impetus for various initiatives to ``unbundle'' PBS--- unbundling PBS in some form is necessary to prevent concentration into a few integrated builder-searchers. Previous work has focused on limiting the power of builders to build blocks by imposing certain constraints on them: see e.g. the recent works of \cite{buterin2022constraint}, \cite{monnot2022pepc}. There have also been studies on the possibility of implementing blockspace futures (see, e.g., \cite{ma2022}), which would effectively partially disintermediate the builder by guaranteeing inclusion for some transactions.

\section{Related Literature}

Although Proposer Builder Separation has only recently become the dominant method of building blocks on Ethereum, research into PBS dynamics is active with several recent developments. \cite{flashbots-order-flow-auctions-and-centralisation} discusses the potential for exclusive order flow to centralize the builder market. \cite{order-flow-auctions-and-centralisation-II} surmised that order flow auctions could potentially alleviate the centralizing effects of private order flow by providing a level playing field for all builders to bid in; however, as we discuss later, our results suggest that order-flow auctions may still have a centralizing effect on PBS because builders with advantages in top of block extraction may come to dominate the auction, resulting in equilibria that look similar to those where private order flow is purchased by a single entity. \cite{Gosselin2023} catalogued several competing order-flow auction designs and outlined their respective advantages and disadvantages. 

A series of articles have provided visibility into the current status of the MEV supply chain. \cite{heimbach2023ethereums} discussed to what degree relays (which forward PBS bids to the proposer) have lived up to their promises about which types of transactions and MEV strategies are allowed in their blocks. \cite{wahrstätter2023time} catalogues the current state of the PBS market, noting which builders submit to which relays, and providing insight into the total revenue from the PBS system so far. \cite{schwarzschilling2023time} discusses how proposers who participate in MEV boost could raise their revenue by delaying block proposal, allowing more bids to come in before choosing a winner. \cite{builder-dominance-and-searcher-dependence} uses proprietary data provided by Titan builder to demonstrate that top builders have more order flow than other builders and that this is a large factor contributing to their dominance in the PBS auctions. 

Recent developments have increased our understanding of CEX/DEX arbitrage also known as \textit{loss-versus-rebalancing} (LVR) or stale order sniping. \cite{milionis2022LVR} proposed the definition for LVR as the loss that the pool incurred relative to a perfect re-balancing portfolio. This quantity can also be thought of as the expected profit of top of block CEX/DEX arbitrage bots. \cite{milionis2023Fees} extended this analysis from continuous time with no fees to discrete time with fees, a much more realistic model. A core finding of this paper was the result that LVR grows cubically in blocktime, with smaller blocktimes leading to lower LVR. 

\section{Background}\label{sec:casestudies}

The easiest way to understand the top-of-block, block-body distinction is to look at the blocks themselves. CEX/DEX arb transactions are easily identifiable since they are large directional trades, typically in the first few slots of the block.

These CEX/DEX arb transactions are usually executed by an MEV bot contract that disproportionately lands transactions in blocks associated with the corresponding builder. For example, block 17195495,\footnote{See, e.g., \url{https://etherscan.io/block/17195495}.}, contains 182 transactions. The first 37 appear to be CEX/DEX arb transactions from an MEV bot with the address 0xA69b\dots e78C.\footnote{0xA69babEF1cA67A37Ffaf7a485DfFF3382056e78C}  These are subjectively large swaps on major pools (Uniswap, Sushiswap etc). For example, the first transaction swaps 4.265 Million USDC for 2168 wETH \footnote{\url{https://eigenphi.io/mev/eigentx/0xca8ec486cb46066b464104c1b91b3e253218dac6e9570408b66962883dcb0f28}} on the Uniswap v3 0.05\% fee pool.\footnote{\url{https://info.uniswap.org/##/pools/0x88e6a0c2ddd26feeb64f039a2c41296fcb3f5640}} The subsequent 36 are also similarly large swaps, each of the order of several hundred wETH.

Note that these CEX/DEX arbitrage transactions are not found on all blocks---for example, the preceding block, 17195494, does not contain such transactions. They typically only appear when there is high volatility in the preceding 12 seconds, and even then, the sizes tend to be much smaller than this selected block in most cases. For example, in the next block 17195496, there is only 1 CEX/DEX arb transaction from the same bot and the volume traded is only 1.2 Million USDC for 600 WETH.%
\footnote{ \url{https://eigenphi.io/mev/eigentx/0x1e82ed1b04d0a0df667f64c7f341a9924c79465e84d9c10d265e988d0818e9c5}}

In the block after that, block 17195497, the same bot has a single CEX/DEX arb transaction, swapping 272k USDT for 138 ETH.\footnote{\url{https://eigenphi.io/mev/eigentx/0x95b1e7dc5f54a5f6ca02be2e17e26e2c73eccac374f88e7451691e88dfcd8fec}} After this transaction, the rest of the block is filled with block-body opportunities. Transactions at indexes 1--4 and 11--14 are sandwich attacks.\footnote{\url{https://eigenphi.io/mev/eigentx/0x95b1e7dc5f54a5f6ca02be2e17e26e2c73eccac374f88e7451691e88dfcd8fec?tab=block}}

Block 17195497 in particular shows that builders can exploit both top-of-block and block-body opportunities in the same block. This is an important aspect of our model, and drives our results.

\section{Data and Empirical Analysis}\label{sec:data}

The driving assumption in the theoretical analysis in Section \ref{sec:theorems} is that some builders are superior at extracting value from the top-of-block opportunities. In this section, we provide empirical evidence for this assumption. We use realized price-volatility on the CEX as a plausibly exogenous instrument that affects top-of-block but not block body opportunities. In particular, price movements on the CEX create arbitage opportunities since DEX prices are by design static until the next block. Large price movements create large arbitrage opportunities, small price movements create small arbitrage opportunities. 


We obtained block-level data from Etherscan for a period corresponding to roughly a month from April 1st, 2023 to May 1st, 2023 (ETH blocks 16950609 to 17150609). We combined this data with detailed price data of ETHUSD from a leading centralized exchange (Binance) in the 12 seconds before each block was built. Price movement in this window gives us a rough estimate of the amount that can be earned through arbitrage with central exchanges for that block.

The merged data surfaces some clear patterns in builder-volatility relationships. Three builders---Manta, Rsync Builder, and Beaver Build---were identified before the analysis as likely to be better at extracting top of block MEV due to rumored connections with High Frequency Trading Firms. We show how realized volatility is related to whether or not one of these three builders constructed the block in Figure \ref{fig:HFT_volatility}.

\begin{figure}[!h]
  \centering
  \includegraphics[width=0.8\textwidth]{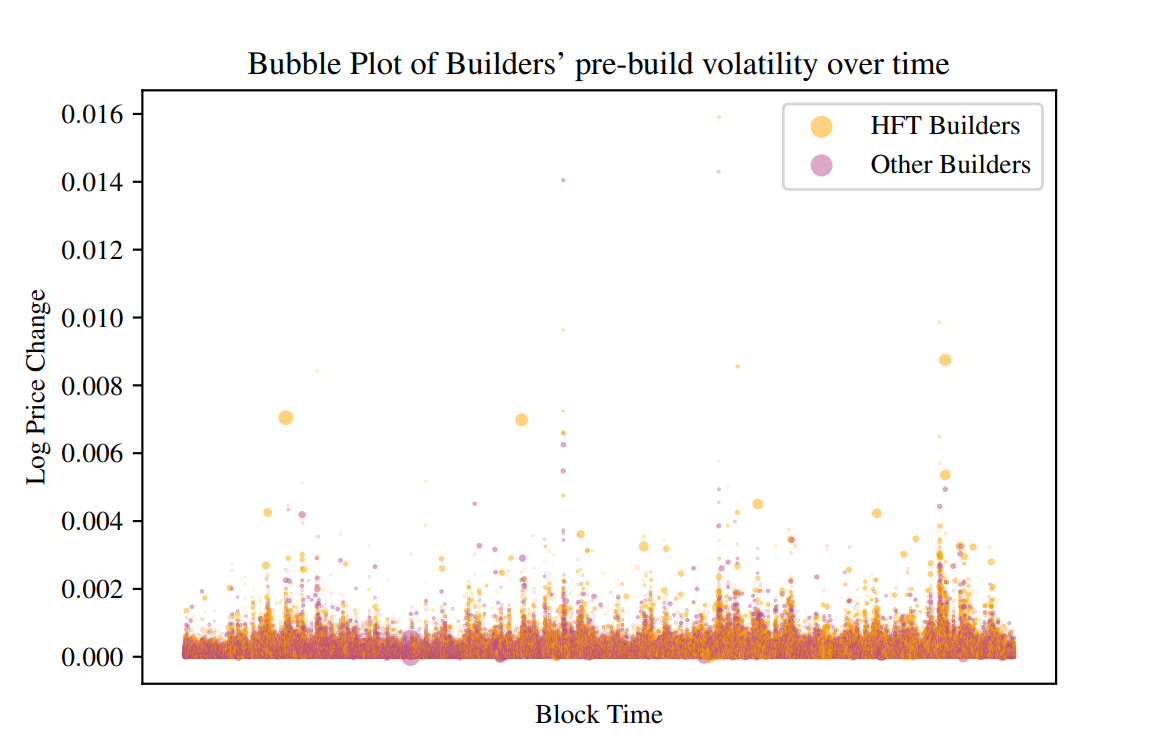}
  \caption{Blocks sorted by block time (x-axis) vs. pre-block volatility measured by log price change (y-axis)}
  \label{fig:HFT_volatility}
\end{figure}


When analyzing the most volatile blocks from each of the HFT traders, there are many large trades with Uniswap v3 pools at the top of each block (as we showed in Section \ref{sec:casestudies}). In general, the larger the realized volatility in the preceding 12 seconds, the larger these trades were. In some cases, when the realized volatility was most extreme, blocks included more than 30 CEX/DEX arbitrage transactions, with many consisting of notional sizes of millions of USD. We document blocks of several notable builders and their respective volatilities in Figure  \ref{fig:all_volatility}.

\begin{figure}[h]
  \centering
  \includegraphics[width=0.8\textwidth]{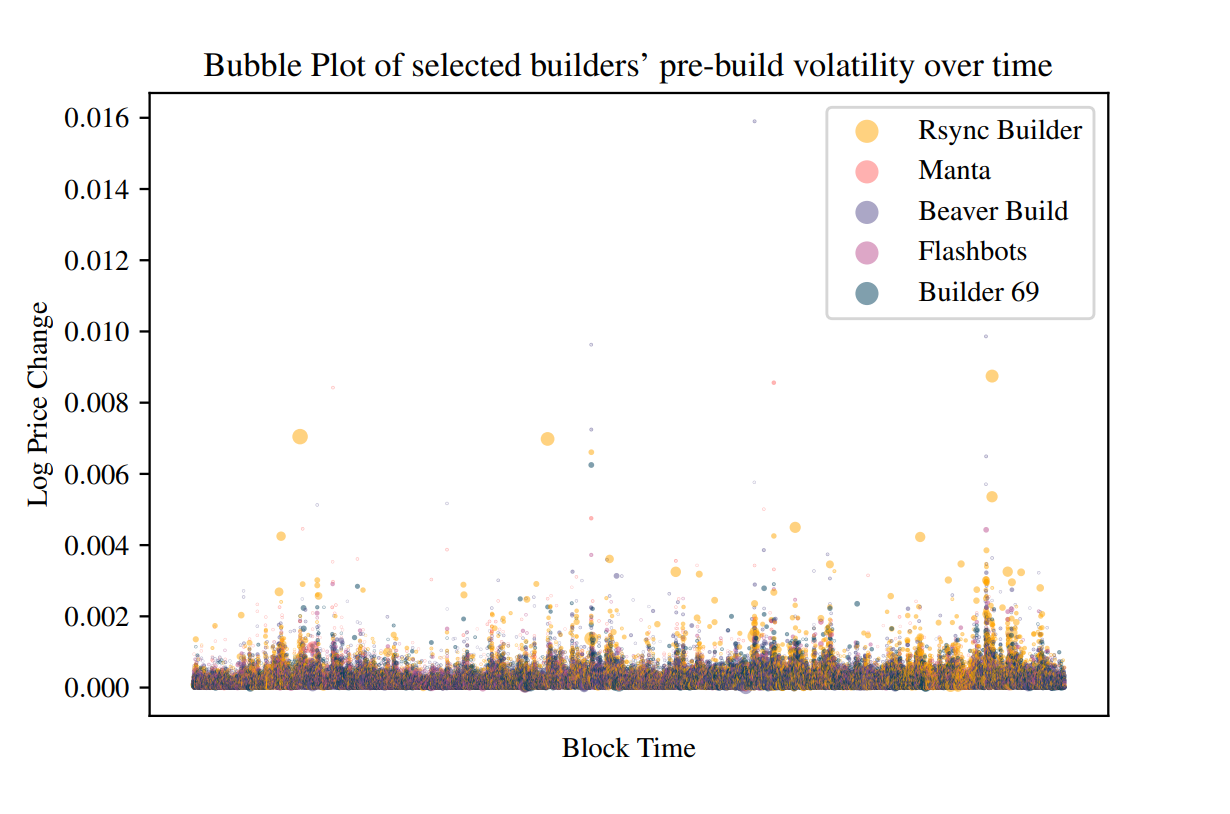}
  \caption{Selected builders' blocks sorted by block time (x-axis) vs. pre-block volatility measured by log price change (y-axis)}
  \label{fig:all_volatility}
\end{figure}

To formalize these findings, we model the relationship between CEX volatility and HFT builders winning the PBS auction. First, we grouped builders into HFT builders (Beaver, Manta, and Rsync) and non-HFT builders (everyone else). We regressed realized volatility on an indicator for whether or not one of these HFT builders won the block using a logistic regression:
\begin{equation*}
P(\text{HFT Builder} = 1 | \text{{Log Price Change}}) = \frac{{1}}{{1 + e^{-(\beta_0 + \beta_1 \cdot \text{{Log Price Change}})}}}
\end{equation*}

\begin{table}[h] 
\centering
\caption{Logistic Regression Results}\label{tab:logistic}
\begin{tabular}{llcccccc}
\toprule
\multicolumn{2}{l}{} & \textbf{Coeff} & \textbf{Std Err} & \textbf{$z$} & \textbf{P $>|z|$} & \textbf{[0.025} & \textbf{0.975]} \\
\midrule
& Log10 Price Change & 2055.151$^{***}$ & (47.584) & 43.190 & 0.000 & 1961.888 & 2148.414 \\
& const & -0.821$^{***}$ & (0.006) & -133.054 & 0.000 & -0.833 & -0.809 \\
\bottomrule
\end{tabular}
\end{table}

We find that the coefficient for the log price change predictor variable is 2055.151, with a standard error of 47.584. The significant positive relationship indicates that as the log price change increases, the odds of HFT builders winning the block also increases. Interpreting the model, when the Log10 Price Change is equal to 0 (i.e., no change) in the period before the block, the log odds of an HFT builder winning the block are -0.821. This corresponds to a probability of 0.306. If the realized volatility was 1\% The probability that an HFT builder won the block was 0.775. When the realized volatility was 2\% the probability that an HFT builder won the block was 0.964. Our analysis therefore comprehensively shows that the likelihood of the builders we preidentified as HFT builders winning the block grows as realized volatility increases. This suggests that these builders are much better than the rest of the field at extracting top-of-block value.

To identify differences between these HFT builder's capabilities, we construct a multinomial logistic regression:

\[
\text{{log}}\left(\frac{P(\text{{Builder}}_i)}{P(\text{{Builder}}_{\text{{ref}}})}\right) = \beta_{0i} + \beta_{1i}(\text{{Log Price Change}})
\]

We restrict analysis to six builders: BeaverBuild, Blocknative, Builder 69, Flashbots,  Manta, and Rsync Builder. Three of these (Beaver, Manta and Rsync) are our aforementioned HFT builders, the remaining three (Blocknative, Builder69 and Flashbots) are high-volume builders that construct a high percentage of the remaining blocks. This model analyzes how the CEX price volatility between blocks impacts the probability of one of these entities becoming the block winner. The resulting model coefficients for each builder in Table \ref{tab:mnlogit} estimate how a unit increase in Log Price Change before a block will impact the log ratio of the probability of that block being won by that particular builder vs. the probability of it being won by a builder in the reference class. While these coefficients are more difficult to interpret than simple logistic model with HFT builders, our findings show  significant relationships between increased volatility before a block and that block being won by a particular builder: the coefficients are positive and significant for our preidentified HFT builders (Beaver, Manta and Rsync) , as one might have expected given our previous results. Conversely, they are significant and negative for the high-volume, non HFT builders (Builder69 and Flashbots).%
\footnote{We note that the coefficient is positive and significant for Blocknative even though Blocknative claims (and industry participants agree) that it is not an HFT builder. A possible reason for this is that Blocknative runs their own relay and presumably collocates their builder with their relay. This could give them a latency advantage which could be influential in winning high volatility blocks. Firms using HFT-type strategies searching for a latency edge may therefore use Blocknative, resulting in the strongly positive coefficient.}

This suggests that these high-volume builders either do not compete in top-of-block extraction activity or at least are substantially less skilled relative to the HFT builders. 

\begin{table}[h] 
\centering
\caption{MNLogit Regression Results}\label{tab:mnlogit}
\begin{tabular}{llcccccc}
\toprule
\multicolumn{2}{l}{} & \textbf{coef} & \textbf{std err} & \textbf{$z$} & \textbf{P $>|z|$} & \textbf{[0.025} & \textbf{0.975]} \\
\midrule
\multicolumn{8}{l}{\textit{ Beaver Build}} \\
& const & -0.4144 & 0.009 & -45.929 & 0.000 & -0.432 & -0.397 \\
& Log10 Price Change & 1386.2014 & 71.403 & 19.414 & 0.000 & 1246.254 & 1526.149 \\
\addlinespace
\multicolumn{8}{l}{\textit{ Blocknative}} \\
& const & -2.4772 & 0.020 & -126.577 & 0.000 & -2.516 & -2.439 \\
& Log10 Price Change & 1629.2443 & 138.304 & 11.780 & 0.000 & 1358.174 & 1900.315 \\
\addlinespace
\multicolumn{8}{l}{\textit{ Builder 69}} \\
& const & 0.0152 & 0.008 & 1.799 & 0.072 & -0.001 & 0.032 \\
& Log10 Price Change & -527.4993 & 75.762 & -6.963 & 0.000 & -675.991 & -379.008 \\
\addlinespace
\multicolumn{8}{l}{\textit{ Flashbots}} \\
& const & -0.4522 & 0.010 & -46.985 & 0.000 & -0.471 & -0.433 \\
& Log10 Price Change & -458.7271 & 86.446 & -5.306 & 0.000 & -628.159 & -289.295 \\
\addlinespace
\multicolumn{8}{l}{\textit{ Manta}} \\
& const & -3.2312 & 0.023 & -137.575 & 0.000 & -3.277 & -3.185 \\
& Log10 Price Change & 3824.6414 & 104.548 & 36.583 & 0.000 & 3619.731 & 4029.551 \\
\addlinespace
\multicolumn{8}{l}{\textit{ Rsync Builder}} \\
& const & -0.6812 & 0.010 & -71.400 & 0.000 & -0.700 & -0.662 \\
& Log10 Price Change & 2093.8362 & 71.075 & 29.459 & 0.000 & 1954.532 & 2233.141 \\
\bottomrule
\end{tabular}
\end{table}

\pagebreak
\section{Model and Theoretical Analysis} \label{sec:theorems}
Having demonstrated the core assumption (that some builders are better than others at extracting value from the top of the block), we turn to a theoretical model of what the downstream effects of this advantage might be as the prevalence of private order flow increases.

We construct and study a simple static model for a single slot. In our model, a block consists of at most $2$ transactions. There is a single available block body transaction which can generate MEV (for example a swap transaction that can be sandwiched). Further, there is a single top-of-block CEX/DEX arbitrage opportunity. There are two builders, $A$ and $B$. Each of these builders competes in the PBS auction to have their block included. In practice, the PBS auction is an English auction--- we will simply consider the standard dominant strategy equilibrium of this auction, i.e., each agent stays in until their value, resulting in the highest value buyer winning at the second highest value.  

We consider two scenarios. Scenario $1$ models the current situation with little/ no private order flow, while scenario $2$ models a setting with private order flow.

\bigskip

\noindent \textbf{Scenario 1}: In this setting the block body transaction is available to both builders, for example as a bundle from a third (unmodeled) searcher. Both builders therefore have the same value for this transaction, equaling the searcher's tip which is paid to the including builder---we will denote this value as $v_T$. At the time of the PBS auction, each builder $x \in \{A,B\}$ also sees their value $v_x$ for the CEX/DEX arb. They then bid in the PBS auction, with the winning bidder's block being included. 

\bigskip

\noindent \textbf{Scenario 2}: In this setting the block body transaction is available for sale at an OFA that runs prior to the PBS auction. The value of the transaction for sale is $v_T$, commonly known among the two bidders. For simplicity we will first assume that this auction runs as a second-price auction, i.e. builders submit bids and the winner (highest bid) pays the second-highest bid. In this setting, the loser of the auction does not have access to the block body transaction. At the time of the PBS auction, each builder $x \in \{A,B\}$ also sees their value $v_x$ for the CEX/DEX arb. They then bid in the PBS transaction, with the winning bidder in this auction having their block included.

\begin{assumption}
    We will assume that for each $x \in \{A,B\}$, $v_x \sim F_x$ where $F_x$ is a CDF on $[0,1]$, and that $v_A \perp v_B$, i.e. $A$ and $B$ are independently drawn.

    Further we assume that $F_A \succ_{\text{FOSD}} F_B$, i.e., builder $A$ is stochastically better at CEX/ DEX arb than builder $B.$
\end{assumption}

Our results in this section show that the outcomes in Scenario 2, i.e., the scenario with OFAs and private order flow, overly advantage builder $A$ over builder $B$ relative to scenario $1$. 

\subsection{Baseline Results}
The basic idea is straightforward and can be easily described in a setting where $v_A$ and $v_B$ are deterministic (or equivalently, $F_A$ and $F_B$ are degenerate distributions). Without loss of generality, assume that $v_A> v_B$. 

\begin{theorem}
    In Scenario 1, suppose that $v_x$ for each of $x \in {A,B}$ is common knowledge among the builders before bidding in the PBS auction. Then the equilibrium of the PBS auction  is that A wins the PBS auction at price $v_T + v_B$.  Their total profit is therefore $v_A - v_B$. 
\end{theorem}

In short, the Theorem asserts that the outcome in Scenario $1$ allocates blockspace efficiently.

\begin{proof}
To see why, note that the block body transaction is available to both builders and has the same value, so the sole differentiation is in terms of their value for the top-of-block (CEX/DEX arb). The value of each bidder $x$ for winning the auction is therefore $v_T + v_x$. In the standard equilibrium of an English auction with complete information, the outcome is efficient with the high value bidder winning at the second highest price. The theorem follows. 
\end{proof}

As a first benchmark to compare this against, suppose in Scenario 2 the builders know their value for the CEX/DEX opportunity before the OFA begins.

\begin{theorem}
    In Scenario 1, suppose that the value $v_x$ for each builder $x \in {A,B}$ for top of block is common knowledge among them before bidding in the OFA. Then the overall outcome of OFA followed by PBS auction is that A wins both auctions at total price $\max (v_T + 2 v_B - v_A, v_B)$. Their total surplus is therefore $\min(2(v_A - v_B), v_A + v_T - v_B)$.  
\end{theorem}

\begin{proof}
The proof follows straightforwardly from backward induction. We can work out the willingness to pay of each party for the transaction in the OFA based on the difference in profit in the PBS auction conditional on who wins the OFA. There are two mutually exclusive, totally exhaustive cases:

\noindent \textbf{Case 1:} $v_A > v_B + v_T$. In this case, note that the winner of the PBS auction is $A$ regardless of who wins the OFA (since we already have that $v_A > v_B$). Therefore $B$ gets a $0$ surplus regardless. As a result, we have that $B$ bids $0$ in the OFA and therefore $A$ wins the transaction. Then, the PBS auction clears at a price of $v_B$ with $A$ winning the block, and the total surplus of $A$ is $v_A + v_T - v_B$. 

\noindent \textbf{Case 2:} $v_A \leq v_B + v_T$. In this case, the winner of the OFA will go ahead and win the PBS (since the value of the transaction $v_T$ plus their own value for the top-of-block opportunity combines will be larger than the competitor's value for the top-of-block opportunity). Note that if $A$ wins the OFA, then they will therefore win the PBS at a price of $v_B$ for a net surplus of $v_A + v_T - v_B$ (and $B$ will make a total surplus of $0$). Conversely, if $B$ wins the OFA, they will win the PBS for a price of $v_A$, with a net surplus of $v_B + v_T - v_A$ (and $A$ will make a total surplus of $0$). 
    
Therefore, $A$'s willingness to pay for the transaction in the OFA is $v_A + v_T - v_B$, whereas $B$ is willing to pay $v_B + v_T - v_A < v_A + v_T - v_B$ (since $v_A > v_B$ by assumption). As a result the OFA will see $A$ winning for a price of $v_B + v_T - v_A$. Combining these (the outcomes of the PBA above and the OFA here) we have the desired result. 
\end{proof}

These results already exhibit the `centralization effects' of private order flow on proposer builder separation: every additional dollar of advantage a builder has in top of block extraction translates into more than a dollars of surplus (for a small advantage, up to two dollars). In short, a builder who is already advantaged has a steeper incentive to invest in improving their advantage.

\subsection{Stochastic Top-of-Block Opportunities}

Our results carry through, \emph{mutatis mutandis}, for a more realistic model where at the time of bidding in the OFA, builders do not know the value of the top of block opportunity. Of course this applies solely to Scenario 2. In this case, builder $x$ at the stage of the OFA bids on the understanding that their top-of-block opportunity will be revealed to them later, and is distributed as $v_x \sim F_x$.  At the conclusion of the OFA, the realized top-of-block opportunity for each builder is revealed  to them, and is modeled as a private value.%
\footnote{It maybe interesting to consider the case where this value is a signal of expected top-of-block value. In this case, we may be in a setting of interdependent values as in \cite{milgrom1982theory}. We leave that study to future work.}

Suppose builder A wins the OFA. In this case, their value for the block is $v_T + v_A$, while builder B's value for the block is $v_B$. Conversely, if builder A loses the OFA, their value for the block is $v_A$ while builder B's value for the block is $V_B + v_T$. 

\begin{theorem}
Builder A's value for the transaction in the OFA, $v_{T,A}$, can be written as:
\begin{align*}
   v_{T,A} &= \int_0^\infty \int_0^{v_A} F_B(v + v_T) - F_B(v- v_T) dv  dF_A(v_A), 
\end{align*}
 with $v_{T,B}$ defined analogously.
\end{theorem}
\begin{proof}
Note that conditional on builder $A$'s value for top of block slot being $v_A$, their interim probability of winning the block is 
\begin{align*}
    x^{\text{win}}_A(v_A) = F_B(v_A + v_T),\\
\intertext{and analogously their probability of winning the block from losing the OFA is} 
    x^{\text{lose}}_A(v_A) = F_B(v_A - v_T). 
\end{align*}
Therefore, by the revenue equivalence theorem (see e.g., Proposition 3.1 of \cite{krishna2009auction}), the  expected surplus of builder $A$ in the PBS auction, conditional on the outcome of the OFA with a value of $V_A$ for the top of the block can be written as
\begin{align*}
&    s^{\text{win}}_A(v_A) = \int_0^{v_A} x^{\text{win}}_A (v) dv = \int_0^{v_A} F_B(v+ v_T) dv,\\
& s^{\text{lose}}_A(v_A) = \int_0^{v_A} x^{\text{lose}}_A (v) dv = \int_0^{v_A} F_B(v- v_T) dv
\end{align*}
Finally, the ex-ante expected surplus from winning can be written as:
\begin{align*}
    & S_A^{\text{win}} = \int_0^\infty s_A^{\text{win}}(v_A)  dF_A(v_A) = \int_0^\infty \int_0^{v_A} F_B(v+ v_T) dv  dF_A(v_A),\\
   \intertext{and expected surplus from losing as, }
   &S_A^{\text{lose}} = \int_0^\infty s_A^{\text{lose}}(v_A)  dF_A(v_A) = \int_0^\infty \int_0^{v_A} F_B(v- v_T) dv  dF_A(v_A).
\end{align*}
Therefore the effective valuation of builder $A$ to win the the transaction in the OFA, $v_{T,A}$ equals $S_A^{\text{win}} - S_A^{\text{lose}}$. Analogously, the valuation of builder $B$ in the transaction in the OFA equals $S_B^{\text{win}} - S_B^{\text{lose}}$.

Note that
\begin{align*}
    v_{T,A} &= S_A^{\text{win}} - S_A^{\text{lose}},\\
     &= \int_0^\infty \int_0^{v_A} F_B(v + v_T) - F_B(v- v_T) dv  dF_A(v_A),\\
 \intertext{and, analogously, } 
 v_{T,B} &=     S_B^{\text{win}} - S_B^{\text{lose}},\\
     &= \int_0^\infty \int_0^{v_B} F_A(v + v_T) - F_A(v- v_T) dv  dF_B(v_A),
\end{align*}
as desired.
\end{proof}

Finally, note that under various assumptions, it can be shown that $v_{T,A} > v_{T,B}$. For example:

\begin{corollary}
    Suppose $v_T$ is small enough so that a Taylor series approximation is appropriate. Then $v_{T,A} \geq v_{T_B}$. 
\end{corollary}
\begin{proof}
    To see this note that 
    \begin{align*}
      v_{T,A}   &= \int_0^\infty \int_0^{v_A} F_B(v + v_T) - F_B(v- v_T) dv  dF_A(v_A),\\
      &\approx \int_0^\infty \int_0^{v_A} 2 v_T f_B(v) dv  dF_A(v_A)\\
      & = 2v_T \int_0^\infty F_B(v_A) f_A( v_A) dv_A.
      \intertext{By an analogous argument,}
      v_{T,B} &\approx 2v_T \int_0^\infty F_A(v_B) f_B( v_B) dv_B.
    \end{align*}
    
Since $F_A \succ_{\text{FOSD}} F_B$, we have that for all $v$, $F_A(v) \leq F_B(v)$. Therefore we have that,
\begin{align*}
    v_{T,A} & \approx 2v_T \int_0^\infty F_B(v_A) f_A( v_A) dv_A\\
    &\geq \int_0^\infty F_A(v_A) f_A( v_A) dv_A &&\text{(since } F_A \succ_{\text{FOSD}} F_B \text{)},\\
    & \geq \int_0^\infty F_A(v_A) f_B( v_A) dv_A &&\text{(since } F_A \succ_{\text{FOSD}} F_B \text{)},\\
    &\approx v_{T,B}.\qedhere
\end{align*}
\end{proof}

Note that this corollary already implies that even though the top of block opportunities are  $v_A$ and $v_B$ are stochastic, builder $A$ \emph{always} wins the OFA, since it expects better (stochastic) top of block opportunities.

Using this, we can compare winning probabilities and builder profit across the two scenarios. We  summarize our results with the following theorem:
\begin{theorem}\label{thm:compare}
    Under Scenario $1$, builder $A$ wins the block with ex-ante probability $\int_0^\infty F_B(v_A) f_A(v_A) dv_A $; whereas under OFAs with private transactions, builder $A$'s winning probability increases to\\ $\int_0^\infty F_B(v_A + v_T) f_A(v_A) dv_A$. 

    Under Scenario $1$, the total expected profit of builder $A$ is $$\int_0^\infty \int_0^{v_A} F_B(v) dv dF_A(v_A).$$

    Under Scenario $2$, the total expected profit of builder $A$ is $$(v_{T,A} - v_{T,B}) + \int_0^\infty \int_0^{v_A} F_B(v +v_T) dv dF_A(v_A).$$
\end{theorem}

\begin{proof}
    To see the first part, note that under scenario 1, builder A wins the block whenever their realized value for the transaction ($v_A$) exceeds builder $B$'s. 

    Under scenario 2, note that by the previous result, builder A always wins the block-body transaction in the OFA. It then therefore wins the PBS auction whenever builder B's value for the top-of-block transaction is at most $v_T$ larger than builder A's value for the same.  The formulae listed  straightforwardly represent the probability of the events described above. 

    The total expected profits then follow from revenue equivalence (see e.g. Proposition 3.1 of \cite{krishna2009auction}). In particular, recall that if a buyer of value wins the object with probability $x(v)$, their expected profit in this auction must be $S(v) = \int_0^v x(v') dv' + S(0).$ The expected profits listed above then follow straightforwardly. 
\end{proof}

By observation, the profits of builder $1$ have gone up: firstly, they make positive profit in the OFA since they are more aggressive in the OFA. Secondly, having won the OFA, they are advantaged in the PBS auction (since they have access to the private transaction to increase their value for the block, and builder B does not). Further results require us to make a functional form assumption on $F_A $ and $F_B$, which we do in the next section.

\subsection{An Analytic Example}
To better understand the effect on surplus etc, we can use the formulas above in an analytic example so that we can do some simple comparative statics. To that end suppose both $v_A$ and $v_B$ are exponentially distributed, with parameter $\lambda_A$ and $\lambda_B$ respectively. By assumption that A is the stronger builder in terms of first order stochastic dominance of top-of-block opportunities, we must have that $\lambda_A < \lambda_B$.

Note that, for each $x \in \{A,B\}$
\begin{align*}
    &F_x(v) = 1- \exp \{-\lambda_x v\},\\
    &f_x(v) = \lambda_x \exp \{-\lambda_x v\}.
\end{align*}

Therefore, substituting in, we have that:
\begin{align*}
v_{T,A} &= \int_0^\infty \int_0^{v_A} F_B(v + v_T) - F_B(v- v_T) dv  dF_A(v_A),\\
&= \int_0^\infty H_B(v_A) dF_A(v_A),\\
\intertext{where}
H_B(v_A) &= \begin{cases}
    \int_{v_T}^{v_A} F_B(v + v_T) - F_B(v- v_T) dv + \int_0^{v_T} F_B(v+v_T) dv & \text{ if } v_A > v_T\\
    \int_0^{v_A} F_B(v+v_T) dv & \text{ o.w.}
\end{cases}
\end{align*}
A mechanical but involved calculation delivers that:
\begin{align*}
    v_{T,A} = \frac{\lambda_A (1 - \exp (-v_T \lambda_B) ) + \lambda_B (1- 
 \exp (-v_T \lambda_A))}{(\lambda_A^2 + \lambda_A \lambda_B)}
 \end{align*}
And analogously $v_{T,B}$. Further it is straightforward to verify that $v_{T,A} > v_{T,B}$ (since $\lambda_A < \lambda_B$ by assumption) as desired. 

Substituting in to the formulas in Theorem \ref{thm:compare}, we have that the probability of $A$ winning rises to
\begin{align*}
    1- \frac{\exp\{- v_T \lambda_B\} \lambda_A}{\lambda_A + \lambda_B} >\frac{\lambda_B}{\lambda_A + \lambda_B}
\end{align*}
where the right hand side is the probability of $A$ winning in Scenario 1. 

Finally, note that under scenario $1$, the total expected profit of Builder $A$ is 
$$\int_0^\infty \int_0^{v_A} F_B(v) dv dF_A(v_A) = \frac{\lambda_B}{\lambda_A(\lambda_A + \lambda_B)}. $$
By comparison, under scenario 2, the total expected profit of Builder $A$ is: 
\begin{align*}
    &(v_{T,A} - v_{T,B}) + \int_0^\infty \int_0^{v_A} F_B(v +v_T) dv dF_A(v_A),\\
    =& \frac{(\lambda_B - \lambda_A)(\lambda_A (1 - \exp (-v_T \lambda_B) ) + \lambda_B (1- 
 \exp (-v_T \lambda_A)))}{\lambda_A \lambda_B (\lambda_A + \lambda_B)} + \frac{\lambda_B + \lambda_A(1 - \exp (-v_T \lambda_B))}{\lambda_A (\lambda_A + \lambda_B)}.\\
 \intertext{Therefore the difference in profit between the two scenarios is: 
}
& \frac{(\lambda_B - \lambda_A)(\lambda_A (1 - \exp (-v_T \lambda_B) ) + \lambda_B (1- 
 \exp (-v_T \lambda_A)))}{\lambda_A \lambda_B (\lambda_A + \lambda_B)} + \frac{ (1 - \exp (-v_T \lambda_B))}{(\lambda_A + \lambda_B)}
\end{align*}
Note that since each of the terms is positive, so is the sum, i.e. builder $A$'s total profit increases in Scenario $2$ relative to scenario $1$. 

These comparative statics are illustrated in Figure \ref{fig:comparative_statics}. We normalize $v_T$ to $1$, and capture the advantage of builder $A$ by varying $\frac{\lambda_A}{\lambda_A + \lambda_B}$ holding $\lambda_A + \lambda_B$ fixed. The smaller the former, the larger is builder $A$'s advantage in top of block extraction. The figure threfore demonstrates how even small advantages can be discontinuously magnified by private OFAs: it is instructive to note that even if the advantaged builder has a small advantage in top-of-block extraction, e.g., $\lambda_B  = \lambda_A + \epsilon$ for $\epsilon$ small, they have a discontinuous jump in their probability of winning the PBS auction in scenario $2$ relative to scenario $1$. This is because even a small advantage in top-of-block extraction leads to the advantaged builder always winning the OFA in Scenario $2$, which in turn gives them a discontinuous advantage in the PBS auction.

\begin{figure}[ht]
\centering
\includegraphics[width = \textwidth]{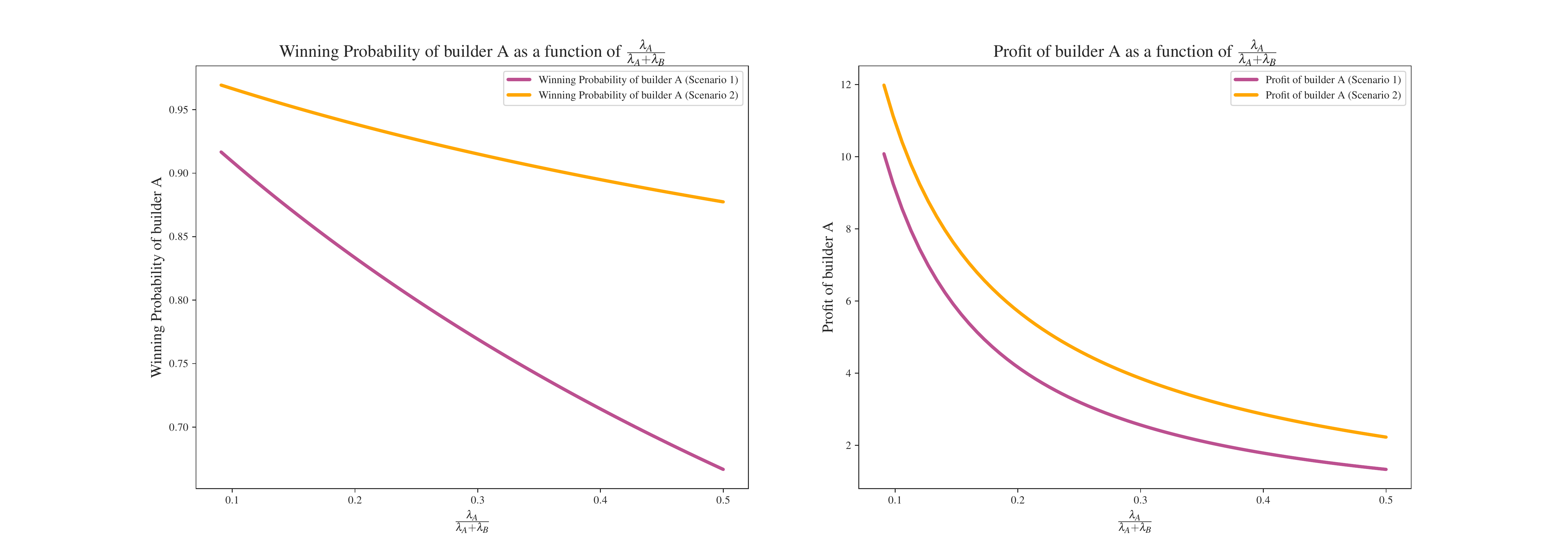}
\caption{How winning probability (left) and expected profit (right) vary across scenarios 1 and 2 as the relative advantage of Builder $A$ varies.}
\label{fig:comparative_statics} 
\end{figure}

\begin{remark}[Bundle Sharing]
    The results of this model apply in a world where builders do not share opportunities with other builders. This practice, known as bundle sharing, is rather common \cite{builder-dominance-and-searcher-dependence}. When a bundle is shared, the bundle tip sets how much of the opportunity's value is shared with the builder and how much is retained for the bundle originator. In a world with low bundle sharing friction, the centralizing effects of private order flow are diminished because the value of running a builder is low (searchers can get almost the same execution without a builder as they can if they do run a builder).

In practice there are several frictions that make bundle sharing less desirable. First, a builder sharing a bundle with another searcher can cause the competing builder to elevate their bid in this PBS auction, making it more likely that the originating builder loses the PBS auction or pays a higher price when he wins. Second, bundles submitted to other builders have higher latency, meaning decisions have to be made earlier in the slot with less information about how the prices of underlying assets will evolve. This latency effect is particularly relevant for top of block arbitrage opportunities which is why most successful top of block searchers also run their own builder. 
\end{remark}

\section{Discussion}\label{sec:conclude}
Our empirical results show that a small group of integrated builder-searchers have a demonstrable advantage in top-of-block extraction capability.   

Our theoretical model then shows that builders with superior top-of-block capabilities are likely to dominate OFAs and subsequently use the private order flow obtained in these OFAs to dominate the PBS auction. Put simply, top-of-block and block-body opportunities are complementary because the block is sold wholesale. Therefore, builders who earn more from the top of the block, will be willing to pay more for private order flow, since they need to win the whole block in order to exercise their top-of-block advantage. This complementarity is a strong centralizing force that threatens to suffocate small builders and upset the currently somewhat pluralistic builder equilibrium. 

Asking order flow originators not to participate in OFAs is futile because it is in their own best interest to do so. Similarly, builders cannot be barred from participating in OFAs. The only solution then is to modify PBS itself. 

Our results suggest that unbundling the PBS auction would be a step in the right direction. By this we mean selling the top of the block and the block-body separately. Implementing such a mechanism would reduce HFT advantage and allow alternative strategies to integrated builder searchers to compete for the right to build blocks. A more fleshed-out proposal in this direction was recently proposed in \cite{barnabe2023} in the form of PEPC-Boost, a specific instantiation of a protocol enforced proposer commitment (PEPC) in which the block is split into designated top-of-block and blockbody and these are auctioned separately.

\newpage

\bibliographystyle{aer}
\bibliography{onchain-auctions}

\end{document}